\documentclass[letterpaper, 10 pt, conference]{ieeeconf}  
\IEEEoverridecommandlockouts                              
\usepackage [latin1]{inputenc}
\usepackage{cancel}
\usepackage{color}
\usepackage[normalem]{ulem}
\usepackage{graphicx, bm, mathrsfs}
\usepackage{cite,array, enumerate}
\usepackage{amssymb}
\usepackage{dsfont, bm}
\usepackage{epstopdf}
\usepackage[cmex10]{amsmath}
\usepackage[noend]{algpseudocode}
\usepackage{algorithmicx,algorithm}

\newtheorem{lem}{Lemma}

\newtheorem{defn}{Definition}
\newtheorem{rem}{Remark}

\title{\LARGE \bf
Average Cost Optimal Control of Stochastic Systems Using Reinforcement Learning}

\author{Jing~Lai, \emph{IEEE Student Member}, Junlin~Xiong,  \emph{IEEE Member}
\thanks{This work was supported by National Natural Science Foundation of China under Grant 61773357}
\thanks{Jing. Lai and Junlin. Xiong are with the Department of Automation, University of Science and Technology of China, Hefei 230026, China. (E-mail: \tt\small lj120@mail.ustc.edu.cn, xiong77@ustc.edu.cn).}
}

\begin{document}

\maketitle
\thispagestyle{empty}
\pagestyle{empty}

\begin{abstract}
This paper addresses the average cost minimization problem for discrete-time systems with multiplicative and additive noises via reinforcement learning.
By using Q-function, we propose an online learning scheme to estimate the kernel matrix of Q-function and to update the control gain using the data along the system trajectories.
The obtained control gain and kernel matrix are proved to converge to the optimal ones.
To implement the proposed learning scheme, an online model-free reinforcement learning algorithm is given, where recursive least squares method is used to estimate the kernel matrix of Q-function.
A numerical example is presented to illustrate the proposed approach.
\end{abstract}

 \section{Introduction}
Reinforcement learning (RL) \cite{sutton2018reinforcement} has been widely applied for solving optimizing problems in uncertain environments.
RL has performed impressively in many challenging tasks including playing Atari games \cite{mnih2015human,pmlr-v48-mniha16}, playing video game Doom \cite{jiang2015self} and has been extensively studied to
control dynamical systems \cite{bucsoniu2018reinforcement,vamvoudakis2017game}.
Recently, extensions to RL-based control schemes for stochastic systems have sprung up as well \cite{recht2019tour,8386658,9029904,9096367}.

As system parameter uncertainties are often modeled as multiplicative noises \cite{boyd1994linear,9115001} and some external disturbances are modeled as additive noises \cite{bian2019continuous}, stochastic systems subjected to multiplicative and additive noises have been studied extensively  by taking advantage of RL methods \cite{todorov2002optimal,9115001,ZHANG20201}.
For systems with additive noises, the authors of \cite{pmlr-v89-abbasi-yadkori19a} evaluated the value function first and then  Q-function. The control policies were updated with respect to the average of all previous Q-function estimations.
Compared with \cite{pmlr-v89-abbasi-yadkori19a}, the authors of \cite{9029904} developed a model-free RL algorithm for the stochastic  linear quadratic regulator (LQR) problem with additive noises in both system states and system measurements.
In \cite{9096367}, RL algorithms were proposed to solve a class of coupled algebraic Riccati equations for linear stochastic dynamics and to minimize the variance of the cost function.
For systems with multiplicative noises, the authors of \cite{WANG20181} presented the value iteration learning algorithm to find the optimal control policy where system matrices were partially needed to implement the algorithm.
In \cite{9115001}, policy iteration were employed to solve a zero-sum dynamic linear quadratic game where systems were subjected to multiplicative noise terms.
For systems suffering from both multiplicative and additive noises, a RL-based model-free control methodology was proposed to solve the discount-optimal control problem for continuous-time linear stochastic systems in \cite{7260102}.
Motivated by \cite{7260102}, the authors of \cite{ZHANG20201} considered the adaptive optimal control problem for stochastic systems with
complementally unmeasurable states.
Based on the separation principle, a data-driven optimal observer and an off-policy data-driven RL algorithm were combined to yield the optimal control policy without the knowledge of the system matrices.

In this paper, the average cost optimal control problem is investigated for a class of discrete-time stochastic systems subject to both multiplicative and additive noises.
This paper aims to find the optimal admissible control policy in the sense of minimizing the average expected cost.
Firstly, with the system matrices known, the control gain and kernel matrix of value function sequences are evaluated in an offline manner.
The obtained sequences  are proved to converge to the optimal control gain and to the solution to the stochastic algebra riccati equation (SARE), respectively.
Then, by using Q-function, we estimate the kernel matrix of Q-function and update the control gain in an online manner without the knowledge of the system matrices.
Finally, an RL-based algorithm is presented to implement the online learning scheme , where recursive least squares method is used to estimate the iterative kernel matrix of Q-function.
Our algorithm removes the assumption that the noises are measurable as needed in \cite{ZHANG20201},\cite{JZP2016noise}, where optimal control problems were considered  for continuous-time systems with multiplicative and additive noises.
Moreover, compared with \cite{7574391,8122054}, the judicious selection of the discount factor has been avoided.
A numerical example is presented to illustrate the obtained results.

\textit{Notation:}
The notation $\mathbb{R}^{n\times m}$ denotes the set of $n\times m$ real matrices.
For matrix or vector $X$, $X^\top$ denotes its transpose, and $\rho(X)$ denotes its spectral radius.
The notation $\|.\|$ denotes the $L_2$ norm for both matrices and vectors.
The trace of a square matrix $X$ is denoted by ${\rm tr}(X)$.
The notation $I$ denotes an identity matrix with appropriate dimensions, \underline{\textbf{1}} is a column vector with all of its elements being 1.
For matrix $X$, $X>0$ (resp. $X\geq 0$) means that $X^\top=X$ and $X$ is positive (resp. positive semi-) definite.
Let $\otimes$ denote the Kronecker product and ${\rm E}$ denote the mathematical expectation.
For symmetric matrix $X\in \mathbb{R}^{n\times n}$,
${\rm {vech}}(X)=[X_{11},X_{12},\ldots,X_{1n},X_{22},X_{23},
\ldots,X_{n-1,n},X_{nn} ]$, ${\rm {vecs}}(X)$ $=[X_{11},2X_{12},\ldots,2X_{1n},X_{22},2X_{23},$
$\ldots,2X_{n-1,n},X_{nn}]$, and both ${\rm {vech}}(X)$ and ${\rm {vech}}(X)$ are in the set of $\mathbb{R}^\frac{n(n+1)}{2}$ .

\section{problem description}
Consider the following discrete-time stochastic system
\begin{equation}  \label{sys1}
x_{k+1}=(A+\sum_{i=1}^{p}\alpha_{ik}A_i)x_k
+(B+\sum_{j=1}^{q}\beta_{jk}B_j)u_k+d_k
\end{equation}
where $x_k\in \mathbb{R}^{n}$ is the system state at time $k$,
$u_k\in \mathbb{R}^{m}$ is the control input,
$x_0$ is the system initial state, which is a Gaussian random vector with zero mean and covariance $X_0\geq0$.
The matrices $A,A_i\in \mathbb{R}^{n\times n}$, $B,B_j\in \mathbb{R}^{n\times m}$ are system matrices.
The system noise sequence $\{(\alpha_{ik},\beta_{jk},d_k): i=1,2,\ldots,p,~j=1,2,\ldots,q,~k=0,1,2,\ldots\}$
is defined on a given complete probability space $(\Omega, \mathcal{F}, \mathcal{P})$.
We assume that there exist linear admissible control policies for system \eqref{sys1}.
Furthermore, assume that: 1)multiplicative noises $\alpha_{ik}$ and $\beta_{jk}$ are scalar Gaussian random variables with zero means and covariances $\bar{\alpha}_i$ and $\bar{\beta}_j$, respectively; 2) additive noise $d_k$ is a Gaussian random vector with zero mean and covariance $D>0$; 3) $x_0,~\alpha_{ik},~\beta_{jk}$ and $d_k$ are mutually independent.

\begin{defn} \label{def1}
System \eqref{sys1} with control input $u_k\equiv0$ is called asymptotically square stationary (ASS) if there exists a matrix $X>0$ such that $\|\underset{k\rightarrow\infty}{\lim}{\rm E}(x_kx_k^{\top})-X\|=0$.
\end{defn}

\begin{defn}\label{def2}
A control policy is admissible if and only if the system followed the control policy is ASS.
\end{defn}

\begin{rem}\label{rem_Exx}
With an admissible control policy $u=Lx$,
we have
\begin{align}
\nonumber &{\rm E}(x_{k+1}x_{k+1}^{\top})=(A+BL){\rm E}(x_kx_k^{\top})(A+BL)^{\top}\\
\nonumber&~+\sum_{i=1}^p\bar{\alpha}_iA_i{\rm E}(x_kx_k^\top) A_i^\top+
\sum_{j=1}^q\bar{\beta}_jB_jL{\rm E}(x_kx_k^\top) L^\top B_j^\top+D.
\end{align}
Hence, ${\rm E}(x_kx_k^{\top})$ is positive definite.
\end{rem}

Define the value function associated with an admissible control policy $u=Lx$ as
\begin{equation}\label{vf3}
J(x_k)={\rm E}\sum_{t=k}^\infty(c(x_t,u_t)-\lambda),
\end{equation}
where the term $c(x_t,u_t)= x_t^{\top}Qx_t+u_t^{\top}Ru_t$ with $Q\geq 0,~R>0$ is the one step cost, and $\lambda$ is the average expected cost
\begin{equation}
\nonumber \lambda=\lim_{N\rightarrow \infty}\frac{1}{N}
{\rm E}\sum_{t=0}^{N}(c(x_t,u_t)),
\end{equation}
i.e., the expected quadratic running cost in the steady state.

This paper aims to find the optimal admissible control policy $u^{\ast}=L^{\ast}x$ in the sense of minimizing the average expected cost $\lambda$.

\begin{lem}\label{lem2}
Let the control policy $u=Lx$ be admissible. Then $\lambda={\rm tr} (PD)$,  where $P>0$ is the unique solution to the following stochastic Lyapunov equation (SLE)
\begin{align}\label{lem2_1}
\nonumber P&=(A+BL)^{\top}P(A+BL)+\sum_{i=1}^p\bar{\alpha}_iA_i^{\top}PA_i\\
&\quad+\sum_{j=1}^q\bar{\beta}_jL^\top B_j^\top PB_jL+Q+L^{\top}RL.
\end{align}
\end{lem}

\begin{proof}
This lemma is a extension of \cite[Section \uppercase\expandafter{\romannumeral3}]{1099303}, and therefore is omitted here.
\end{proof}

Based on equation \eqref{vf3}, one has a Bellman equation for the value function $J(.)$
\begin{align}\label{bell}
J(x_k)={\rm E}\big(c(x_k,u_k)\big)-\lambda+ J(x_{k+1}).
\end{align}

\begin{lem}\label{lem_qua}
Let the control policy $u=Lx$ be admissible.
Then, the value function \eqref{vf3} can be written as
\begin{equation}\label{vf4}
J(x_k)={\rm E}(x_k^{\top}Px_k)+\bar{s},
\end{equation}
where $P>0$ is the unique solution to SLE \eqref{lem2_1} and the term $\bar{s}$ is a const.
\end{lem}
\begin{proof}
Without loss of generality, suppose $J(x_k)={\rm E}(x_k^{\top}\bar{P}x_k)+s_k$, where $s_k$ is independent on $x_k$.
Let the control policy $u=Lx$ be admissible, $x_{k+1}$ is given by
\begin{equation}
\nonumber x_{k+1}=(A+BL)x_k+\sum_{i=1}^p{\alpha}_{ik}A_ix_k+
\sum_{j=1}^q\beta_{jk}B_jLx_k+d_k.
\end{equation}
Therefore,
\begin{align}
\nonumber {\rm E}&(x_{k+1}^{\top}\bar{P}x_{k+1})={\rm E}\big(x_k^{\top}\big((A+BL)^{\top}\bar{P}(A+BL)\\
\nonumber&+\sum_{i=1}^p\bar{\alpha}_iA_i^{\top}\bar{P}A_i
+\sum_{j=1}^q\bar{\beta}_jL^{\top}B_j^\top\bar{P}B_jL
\big)x_k\big)+ {\rm tr}(\bar{P}D).
\end{align}
Based on the Bellman equation for the value function $J(.)$, one has
\begin{align}
\nonumber &J(x_k)-J(x_{k+1})={\rm E}\big( x_k^\top(Q+L^\top R L) x_k\big)-\lambda\\
\nonumber &={\rm E}(x_k^{\top}\bar{P}x_k)+s_k-{\rm E}(x_{k+1}^{\top}\bar{P}x_{k+1})+s_{k+1}\\
\nonumber &={\rm E}\big(x_k^{\top}\big(\bar{P}-(A+BL)^{\top}\bar{P}(A+BL)
-\sum_{i=1}^p\bar{\alpha}_iA_i^{\top}\bar{P}A_i\\
\nonumber &-\sum_{j=1}^q\bar{\beta}_jL^{\top}B_j^\top\bar{P}B_jL
\big)x_k\big)+s_k-s_{k+1}- {\rm tr}(\bar{P}D).
\end{align}
By matching terms, one has
\begin{align}
\nonumber \bar{P}&=(A+BL)^{\top}\bar{P}(A+BL)
+\sum_{i=1}^p\bar{\alpha}_iA_i^{\top}\bar{P}A_i\\
\nonumber&\quad+\sum_{j=1}^q\bar{\beta}_jL^{\top}B_j^\top\bar{P}B_jL
+Q+L^{\top}RL.
\end{align}
Based on the admissibility of the control policy $u=Lx$ and Lemma \ref{lem2}, we obtain $\bar{P}=P$, which means that the kernel matrix of the average expected cost $\lambda$ is equal to the kernel matrix of the value function $J(.)$ in face of the same admissible control policy.
Thus, $\lambda={\rm tr}(PD)={\rm tr}(\bar{P}D)$ and $s_k=s_{k+1}$.
This completes the proof.
\end{proof}

\begin{rem}
The authors of \cite{pmlr-v89-abbasi-yadkori19a,9029904} considered the value function \eqref{vf3} to be quadratic functions without the constant $\bar{s}$.
Here we give a particular proof to show that the constant may exist independent of $x_k$.
Note that one can also prove the existence of the term $\bar{s}$ by the method in the proof of \cite[Lemma 1]{WANG20181}.
\end{rem}

From Lemma \ref{lem2} and Lemma \ref{lem_qua}, one sees that minimizing the average expected cost $\lambda$ is equivalent to minimizing the value function $J(.)$ in the sense that the kernel matrix $P$ is equal in face of the same admissible control policy.
Hence, we convert the original problem into finding the optimal admissible control policy $u^{\ast}=L^{\ast}x$ in the sense of minimizing the value function $J(.)$.

Putting the value function \eqref{vf4} into  equation \eqref{bell},
one obtains the Bellman equation in terms of
the kernel matrix $P$ of the value function
\begin{align}\label{bell2}
 {\rm E}(x_k^{\top}Px_k)&={\rm E}\big(c(x_k,u_k)\big)-\lambda
 +{\rm E}(x_{k+1}^{\top}Px_{k+1}).
\end{align}
The stochastic LQR problem can be
solved based on a stochastic algebra Riccati equation (SARE) given in the following lemma.

\begin{lem}\label{lem_op}
The optimal control gain for the stochastic LQR problem is
\begin{equation}\label{lem_op_1}
L^{\ast}=-(R+ B^{\top}P^{\ast}B+ \sum_{j=1}^q\bar{\beta}_jB_j^\top P^{\ast}B_j)^{-1} B^{\top}P^{\ast}A,
\end{equation}
and $P^{\ast}>0$ is the unique solution to the following SARE
\begin{align}\label{lem_op_2}
\nonumber  P^{\ast}&=Q+ A^{\top}P^{\ast}A+\sum_{i=1}^p\bar{\alpha}_iA_i^{\top}P^{\ast}A_i
-A^{\top}P^{\ast}B\\
&\quad \times(R+ B^{\top}P^{\ast}B+ \sum_{j=1}^q\bar{\beta}_jB_j^\top P^{\ast}B_j)^{-1}
        B^{\top}P^{\ast}A.
\end{align}
Hence, the minimal average expected cost is given as
\begin{align}\label{lem_op_3}
\lambda^\ast={\rm tr}(P^\ast D).
\end{align}
\end{lem}

\begin{proof}
The proof for Lemma \ref{lem_op} is based on the first order necessary condition, and therefore is omitted here.
\end{proof}

\section{model-based scheme to solve stochastic LQR}

In this section, a model-based iterative scheme is provided to find the minimum $P^{\ast}$ and the optimal control gain $L^{\ast}$.
A set of control gains are evaluated in an off-line manner in Lemma \ref{lem3} requiring  complete
knowledge of the system matrices.
\begin{lem}\label{lem3}
Let the initial control gain $L^{(0)}$ be admissible. Consider the two sequences $\{P^{(\tau)}\}_{\tau=0}^\infty$ and $\{L^{(\tau)}\}_{\tau=1}^\infty$ obtained via solving
\begin{align}\label{lem3_1}
\nonumber &P^{(\tau)}= (A+BL^{(\tau)})^{\top}P^{(\tau)}(A+BL^{(\tau)})
+\sum_{i=1}^p\bar{\alpha}_iA_i^{\top}P^{(\tau)}A_i\\
 &+\sum_{j=1}^q\bar{\beta}_j(L^{(\tau)})^\top B_j^\top P^{(\tau)}B_jL^{(\tau)}+(L^{(\tau)})^{\top}RL^{(\tau)}+Q
 \end{align}
and
 \begin{align}\label{lem3_2}
 L^{(\tau+1)}=-\big(R+ B^{\top}P^{(\tau)}B+ \sum_{j=1}^q\bar{\beta}_jB_j^\top P^{(\tau)}B_j\big)^{-1} B^{\top}P^{(\tau)}A.
\end{align}
Then, the following statements hold:
\begin{enumerate}
  \item $P^{\ast}\leq P^{(\tau+1)}\leq P^{(\tau)}$;
  \item $\underset{\tau\rightarrow \infty}{\lim}P^{(\tau)}=P^{\ast}$,
  $\underset{\tau\rightarrow \infty}{\lim}L^{(\tau)}=L^{\ast}$, where $P^{\ast}$ is the unique solution to SARE \eqref{lem_op_2}
 and $L^{\ast}$ is computed as \eqref{lem_op_1};
  \item $L^{\ast}$ and $L^{(\tau)}$ are admissible.
\end{enumerate}
\end{lem}
\begin{proof}
The proof of Lemma \ref{lem3} is a simplified version of \cite[Lemma 6]{lai2020model}, by letting the discount factor $\gamma=1$.
\end{proof}

\begin{rem}
Different from \cite{7574391,8122054,lai2020model}, the need for judicious selection the discount factor has been avoided and the iterative control gains are admissible naturally with the condition that the initial control gain is admissible.
\end{rem}

\section{Model-free Scheme to Solve Stochastic LQR}
In this section, in order to avoid resorting to the
system matrices, a new model-free learning scheme is proposed to solve the stochastic LQR problem by using Q-function.

Based on Bellman equation \eqref{bell},
define a Q-function as
\begin{equation}\label{eq14}
Q(x_k,\eta_k)={\rm E}\big(c(x_k,\eta_k)\big)-\lambda  + J(x_{k+1}),
\end{equation}
where $\eta_k$ is an arbitrary control input at time $k$ and the control policy $u=Lx$ is followed from time $k+1$ onwards. If $\eta_k=u_k$, one knows that
\begin{equation}\label{eq15}
Q(x_k,u_k)=J(x_k).
\end{equation}

Based on Lemma \ref{lem_qua} and system \eqref{sys1}, the above Q-function becomes
\begin{align}\label{eq57}
\nonumber Q(x_k,u_k)&={\rm E}\big(\begin{bmatrix}x_k\\
                        u_k
         \end{bmatrix}^{\top}
\begin{bmatrix}  H_{xx}&H_{xu}\\
                             H_{ux}&H_{uu}
         \end{bmatrix}
         \begin{bmatrix}x_k\\
                        u_k
         \end{bmatrix}\big)+\bar{s}\\
&\triangleq   {\rm E}\big(\begin{bmatrix}x_k\\
                        u_k
         \end{bmatrix}^{\top}
H
         \begin{bmatrix}x_k\\
                        u_k
         \end{bmatrix}\big)+\bar{s}       ,
\end{align}
with
\begin{align}
\nonumber H_{xx}&=Q+ A^{\top}PA+ \sum_{i=1}^p\bar{\alpha}_iA_i^{\top}PA_i\\
\nonumber H_{xu}&=  A^{\top}PB=H_{ux}^\top\\
\nonumber  H_{uu}&=R+ B^{\top}PB+ \sum_{j=1}^q\bar{\beta}_jB_j^\top PB_j.
\end{align}

Denote the optimal Q-function as \cite{sutton2018reinforcement}
\begin{equation}
\nonumber Q^{\ast}(x_k,u_k)={\rm E}\big(c(x_k,u_k)\big)-\lambda^\ast+ J^{\ast}(x_{k+1}).
\end{equation}
Through solving
$\frac{\partial Q^{\ast}(x_k,u_k)}{\partial u_k}=0$, one obtains the optimal control gain
\begin{equation}\label{eq59}
L^{\ast}=-(H^{\ast}_{uu})^{-1}H^{\ast}_{ux},
\end{equation}
where $H^{\ast}_{uu}=R+ B^{\top}P^{\ast}B+ \sum_{j=1}^q\bar{\beta}_jB_j^\top P^{\ast}B_j$,
$H_{ux}=B^{\top}P^{\ast}A$, and $P^{\ast}$ satisfies SARE \eqref{lem_op_2}.
Therefore, from \eqref{eq59} we know that the optimal control gain can be obtained  by finding the optimal kernel matrix of Q-function.

From \eqref{vf4}, \eqref{eq15}, \eqref{eq57} and the positive definiteness  of ${\rm E}(x_kx_k^{\top})$ in Remark 1, we get
\begin{equation}\label{eq60}
P=      \begin{bmatrix}I\\
                        L
        \end{bmatrix}^{\top}
        H
        \begin{bmatrix}I\\
                        L
        \end{bmatrix}.
\end{equation}

Hence, we obtain the Bellman function for Q-function
\begin{align}\label{eq16}
\nonumber Q(x_k,u_k)&={\rm E}\big(c(x_k,u_k)\big) -{\rm {tr}}( H\begin{bmatrix}I\\
                   L
    \end{bmatrix} D
    \begin{bmatrix}I\\
                   L
    \end{bmatrix}^{\top})\\
&\quad +  Q(x_{k+1},u_{k+1})
\end{align}
based on Lemma \ref{lem2}, \eqref{eq14}, \eqref{eq15} and \eqref{eq60}.
Furthermore, substituting  \eqref{eq57} into \eqref{eq16}, one has the Bellman equation in terms of the kernel matrix $H$ of Q-function
\begin{align}\label{eq16_2}
\nonumber {\rm E}\big(\begin{bmatrix}x_k\\
                        u_k
         \end{bmatrix}^{\top}
H
         \begin{bmatrix}  x_k\\
                          u_k
         \end{bmatrix}\big)&={\rm E}\big(c(x_k,u_k)\big)
         - {\rm {tr}}\big( H\begin{bmatrix}I\\
                   L
    \end{bmatrix} D
    \begin{bmatrix}I\\
                   L
    \end{bmatrix}^{\top}\big)\\
 &\quad + {\rm E}\big(
         \begin{bmatrix}x_{k+1}\\
                        u_{k+1}
         \end{bmatrix}^{\top}
         H
         \begin{bmatrix}  x_{k+1}\\
                          u_{k+1}
         \end{bmatrix} \big)    .
\end{align}

In the following lemma, we give a model-free learning scheme, which is inspired by equations \eqref{eq16_2} and \eqref{eq59}, to learn the optimal control policy online where no system matrices are needed.
\begin{lem}\label{lem4}
Let the initial control gain $L^{(0)}$ be admissible. Consider the two sequences  $\{H^{(\tau)}\}_{\tau=0}^\infty$ and $\{L^{(\tau)}\}_{\tau=1}^\infty$ obtained through the following two steps:
\begin{enumerate}
  \item estimate $H^{(\tau)}$ by solving
\begin{align}\label{on_pi_q_1}
\nonumber &{\rm E}\big(\begin{bmatrix}x_k\\
                        u^{(\tau)}_k
         \end{bmatrix}^{\top}
H^{(\tau)}
         \begin{bmatrix}  x_k\\
                          u^{(\tau)}_k
         \end{bmatrix}\big)\\
\nonumber&={\rm E} \big(c(x_k,u^{(\tau)}_k)\big)
- {\rm {tr}}\big( H^{(\tau)}\begin{bmatrix}I\\
                   L^{(\tau)}
    \end{bmatrix} D
    \begin{bmatrix}I\\
                   L^{(\tau)}
    \end{bmatrix}^{\top}\big)\\
&\quad  + {\rm E}\big(
         \begin{bmatrix}x_{k+1}\\
                        u^{(\tau)}_{k+1}
         \end{bmatrix}^{\top}
         H^{(\tau)}
         \begin{bmatrix}  x_{k+1}\\
                          u^{(\tau)}_{k+1}
         \end{bmatrix}\big);
\end{align}
  \item update the control gain through
\begin{align}\label{on_pi_q_2}
L^{(\tau+1)}=-(H^{(\tau)}_{uu})^{-1}H^{(\tau)}_{ux}.
\end{align}
\end{enumerate}
Then, $\underset{\tau\rightarrow\infty}{\lim} H^{(\tau)}=H^{\ast}=\begin{bmatrix}  H_{xx}^\ast&H_{xu}^\ast\\
                             H_{ux}^\ast&H_{uu}^\ast
         \end{bmatrix}$ and $\underset{\tau\rightarrow\infty}{\lim} L^{(\tau)}=L^{\ast}$, where
\begin{align}
\nonumber H_{xx}^\ast&=Q+ A^{\top}P^\ast A+ \sum_{i=1}^p\bar{\alpha}_iA_i^{\top}P^\ast A_i\\
\nonumber H_{xu}^\ast&= A^{\top}P^\ast B=(H_{ux}^\ast)^\top\\
\nonumber  H_{uu}^\ast&=R+ B^{\top}P^\ast B+ \sum_{j=1}^q\bar{\beta}_jB_j^\top P^\ast B_j
\end{align}
with $P^\ast$ being the unique solution to SARE \eqref{lem_op_2}.
\end{lem}

\begin{proof}
Using $c(x_k,u^{(\tau)}_k)=x_k^{\top}Qx_k+(u^{(\tau)}_k)^{\top}Ru^{(\tau)}_k$
and $u_k^{(\tau)}=L^{(\tau)}x_k$, equation \eqref{on_pi_q_1} becomes
\begin{align}
\nonumber {\rm E}&
\big(x_k^{\top} \begin{bmatrix}I\\
                        L^{(\tau)}
         \end{bmatrix}^{\top}
H^{(\tau)}
         \begin{bmatrix}  I\\
                          L^{(\tau)}
         \end{bmatrix} x_k\big)
         ={\rm E} (x_k^{\top}Qx_k+(u^{(\tau)}_k)^{\top}Ru^{(\tau)}_k)\\
\nonumber&\quad- {\rm {tr}}\big( H^{(\tau)}\begin{bmatrix}I\\
                   L^{(\tau)}
    \end{bmatrix} D
    \begin{bmatrix}I\\
                   L^{(\tau)}
    \end{bmatrix}^{\top}\big)\\
\nonumber &\quad + {\rm E}\big(
x_{k+1}^{\top} \begin{bmatrix}I\\
                        L^{(\tau)}
         \end{bmatrix}^{\top}
H^{(\tau)}
         \begin{bmatrix}  I\\
                          L^{(\tau)}
         \end{bmatrix} x_{k+1}\big) .
\end{align}
The above equation can be written as
\begin{align}\label{lem9_2}
\nonumber {\rm E}(x_k^{\top}P^{(\tau)}x_k)
&={\rm E} (x_k^{\top}Qx_k+(u^{(\tau)}_k)^{\top}Ru^{(\tau)}_k)- {\rm {tr}}( P^{(\tau)}D)\\
&\quad + {\rm E}(x_{k+1}^{\top}P^{(\tau)}x_{k+1})
\end{align}
according to equation \eqref{eq60}.
Substituting system \eqref{sys1} and $u_k^{(\tau)}=L^{(\tau)}x_k$ into \eqref{lem9_2}, one gets
\begin{align}
\nonumber &{\rm E}(x_k^{\top}P^{(\tau)}x_k)\\
\nonumber&={\rm E}\Big(x_k^{\top}\big(
 (A+BL^{(\tau)})^{\top}P^{(\tau)}(A+BL^{(\tau)})+\sum_{i=1}^p\bar{\alpha}_iA_i^{\top}PA_i\\
& +(L^{(\tau)})^{\top}(\sum_{j=1}^q\bar{\beta}_j B_j^\top
 PB_j+R)L^{(\tau)}+Q\big)x_k\Big).
\nonumber\end{align}
From Remark \ref{rem_Exx}, one knows that ${\rm E}(x_k x_k^\top)>0$. Hence, we  conclude that equation \eqref{on_pi_q_1} is equivalent to equation \eqref{lem3_1}.
Moreover, the equation \eqref{on_pi_q_2} is equivalent to equation \eqref{lem3_2} according to equations \eqref{eq57} and \eqref{eq59}.
From Lemma \ref{lem3}, we know that $\underset{\tau\rightarrow\infty}{\lim} P^{(\tau)}=P^{\ast}$ and $\underset{\tau\rightarrow\infty}{\lim} L^{(\tau)}=L^{\ast}$. Hence, we conclude that $\underset{\tau\rightarrow\infty}{\lim} H^{(\tau)}=H^{\ast}$ and $\underset{\tau\rightarrow\infty}{\lim} L^{(\tau)}=L^{\ast}$.
This completes the proof.
\end{proof}

\begin{rem}\label{rem5}
Compared with the iterative scheme in
Lemma \ref{lem3}, the scheme in Lemma \ref{lem4} evaluates the iterative matrix $H^{(\tau)}$ and updates the control gain in an online manner only using system states and control inputs without requiring any knowledge of the system matrices.
\end{rem}

\section{implementation of online model-free RL  Scheme}
In this section, recursive least squares (RLS) \cite{goodwin2014adaptive}
is leveraged to estimate the kernel matrix $H^{(\tau)}$ of Q-function.
The implementation of the model-free learning scheme in Lemma \ref{lem4} is given in Algorithm \ref{A2}.




Vectorize the equation \eqref{on_pi_q_1} as
\begin{align}
\nonumber\Big({\rm E}
\big(\phi(z^{(\tau)}_k)\big)\Big)^{\top}& {\rm {vecs}}(H^{(\tau)})=
 \Big({\rm E}
\big(\phi(z^{(\tau)}_{k+1})\big)\Big)^{\top}{\rm {vecs}}(H^{(\tau)})\\
\nonumber &\quad +{\rm E}\big(c(z^{(\tau)}_k)\big)-
\big( {\rm {vech}}(\kappa^{(\tau)})\big)^\top{\rm {vecs}}(H^{(\tau)}),
\end{align}
where
\begin{align}
 \nonumber  z^{(\tau)}_k=[x_k^{\top}~(u^{(\tau)}_k)^{\top}]^{\top}\in \mathbb{R}^{n+m=r} &,~ \phi(z^{(\tau)}_k)={\rm {vech}}\big(z^{(\tau)}_k(z^{(\tau)}_k)^{\top}\big),
\end{align}
 and
\begin{equation}
\nonumber \kappa^{(\tau)}=\begin{bmatrix}I\\
                   L^{(\tau)}
    \end{bmatrix} D
    \begin{bmatrix}I\\
                   L^{(\tau)}
    \end{bmatrix}^{\top}.
\end{equation}

The iterative kernel matrix $H^{(\tau)}$ of Q-function is estimated using data which are generated under the admissible control policy $u_k=L^{(\tau)}x_k$ for $N$ time steps.
Ignoring the expectation operation, the general bath least squares (BLS) estimator of $H^{(\tau)}$ is given by  \cite{5152964}:
\begin{align}\label{eq73}
\nonumber {\rm {vecs}}(H^{(\tau)})&=\big((\Phi^{(\tau)})^{\top}(\Phi^{(\tau)}- \bar{\Phi}^{(\tau)}+ K^{(\tau)})\big)^{-1}\\
&\quad \times(\Phi^{(\tau)})^{\top}\Upsilon^{(\tau)},
\end{align}
where $\Phi^{(\tau)},~\bar{\Phi}^{(\tau)}\in \mathbb{R}^{N\times \frac{r(r+1)}{2}}$ and  $\Upsilon^{(\tau)}\in \mathbb{R}^{N}$ are the data matrices constructed by
\begin{align}
\nonumber &\Phi^{(\tau)}=\begin{bmatrix}
\phi(z^{(\tau)}_0)&\phi(z^{(\tau)}_1)&\cdots&\phi(z^{(\tau)}_{N-1})
\end{bmatrix}^{\top}\\
\nonumber &\bar{\Phi}^{(\tau)}=\begin{bmatrix}
\phi(z^{(\tau)}_1)&\phi(z^{(\tau)}_2)&\cdots&\phi(z^{(\tau)}_{N})
\end{bmatrix}^{\top}\\
\nonumber &\Upsilon^{(\tau)}=\begin{bmatrix}
c(z^{(\tau)}_0)&c(z^{(\tau)}_1)&\cdots&c(z^{(\tau)}_{N-1})
\end{bmatrix}^{\top},
\end{align}
and $K^{(\tau)}$ is a $N\times \frac{r(r+1)}{2}$ matrix whose rows are vectors
${\rm {vech}}(\kappa^{(\tau)})$.
Note that control input $u^{(\tau)}_k$ is dependent on system state $x_k$ linearly.
Generally, it is necessary to add a probing noise to control input $u^{(\tau)}_k$ to guarantee that the persistency of excitation condition holds \cite{bucsoniu2018reinforcement}.

\begin{rem}
If the additive noise covariance matrix $D$ is unknown in practice, one can leverage the empirical average cost $\bar{\lambda}=\frac{1}{N}\sum_{k=0}^{N-1}c(z_k^{(\tau)})$ to approximate  the average expected cost $\lambda^{(\tau)}\triangleq {\rm {tr}}( P^{(\tau)}D)=\big( {\rm {vech}}(\kappa^{(\tau)})\big)^\top{\rm {vecs}}(H^{(\tau)})$. In this case, the BLS estimator of $H^{(\tau)}$ is given as following
\begin{align}
\nonumber {\rm {vecs}}(H^{(\tau)})&=\big((\Phi^{(\tau)})^{\top}(\Phi^{(\tau)}- \bar{\Phi}^{(\tau)})\big)^{-1}
(\Phi^{(\tau)})^{\top}(\Upsilon^{(\tau)}-\bar{\lambda}\underline{\textbf{1}}).
\end{align}
\end{rem}

However, the condition of full rank in \eqref{eq73} may not be satisfied until a sufficient number of states and inputs has been collected.
Moreover, with the dimension of state and input increasing, the inverse operation has higher computational complexity and lower accuracy.
According to the derivation in \cite[Section 6]{lagoudakis2003least}, we use RLS to compute the inverses recursively, which is more time efficient.
Finally, the implementation of the model-free learning scheme is given in Algorithm 1.

\begin{algorithm}[htb]
\caption{Online Model-Free RL} 
\label{A2}
\hspace*{0.02in} {\bf Input:} 
Admissible control gain $L^{(0)}$, initial state covariance matrix $X_0$, additive noise covariance matrix $D$, roll out length $N$, variance $\sigma_u^2$, large positive constant $\varpi$, maximum number of iterations $\tau_{max}$,  convergence tolerance $\varepsilon$\\
\hspace*{0.02in} {\bf Output:} 
The estimated optimal control gain $\hat{L}$
\begin{algorithmic}[1]
\For{$\tau=0:\tau_{max}$} 
    \State \textbf{Policy Evaluation:}
¡¡¡¡\State  Sample $x_0$ from a Gaussian distribution with zeros
    \State  mean and covariance $X_0$. Let $\xi_0=\varpi I$, $\psi_0=0$
    \For{$k=0:N-1$}
    \State Input $u_k=L^{(\tau)}x_k+e_k$ into system \eqref{sys1} to obtain \State $x_{k+1}$, where $e_k$ sampling from  a
     Gaussian distri-
    \State bution with zeros mean and covariance $\sigma_u^2$. Then \State$u_{k+1}=L^{(\tau)}x_{k+1}$
    \State $
\xi_{k+1} = \xi_k-\frac{\xi_k\phi(z_k^{(\tau)})\big( \phi(z_k^{(\tau)}) -\phi(z_{k+1}^{(\tau)})+\kappa^{(\tau)}\big)^\top\xi_k}{1+\big( \phi(z_k^{(\tau)}) -\phi(z_{k+1}^{(\tau)})+\kappa^{(\tau)}\big)^\top\xi_k\phi(z_k^{(\tau)})}
    $
    \State $\psi_{k+1}=\psi_k+\phi(z_k^{(\tau)})c(x_k,u^{(\tau)}_k)$
    \EndFor\State\textbf{endfor}
    \State ${\rm {vecs}}(H^{(\tau)})=\xi_N \psi_N$
    \State \textbf{Policy Improvement:}
    \State $L^{(\tau+1)}=-(H^{(\tau)}_{uu})^{-1}H^{(\tau)}_{ux}$
¡¡¡¡\If{$\|L^{(\tau+1)}-L^{(\tau)}\|<\varepsilon$}
        \State Break
¡¡¡¡\EndIf\State\textbf{endif}
\EndFor\State\textbf{endfor}
\State $\hat{L}=L^{(\tau+1)}$
\end{algorithmic}
\end{algorithm}

\begin{rem}\label{rem8}
In \cite{WANG20181}, the system matrices were partially needed in the implementation of the value iteration algorithm.
In this paper, we have no requirement of the knowledge of the system matrices.
In \cite{JZP2016noise} and \cite{ZHANG20201}, the authors assumed that the terms with noises were measurable.
Here we remove this assumption.
Furthermore, the need for judicious selection the discount factor has been avoid, which is necessary in \cite{7574391,8122054}.
\end{rem}

\section{numerical example}
In this section, a numerical example is presented to evaluate the proposed method. Consider the following open-loop asymptotically square stationary linear discrete-time system:
\begin{align}\label{exam}
\nonumber x_{k+1}&=\begin{bmatrix}
                        0.8672 & 0.0519 & 0.1028\\
                        0.0519 & 0.7576 & 0.0475\\
                        0.1028 & 0.0475 & 0.7681
                    \end{bmatrix}x_k+I_3 u_k\\
\nonumber &~ + \alpha_{1k}
                        \begin{bmatrix}
                        0 & -1 & 0\\
                        -1 & 0 & 0\\
                        0 & 0 & 0
                        \end{bmatrix}x_k
                        + \alpha_{2k}
                        \begin{bmatrix}
                        0 & 0 & -1\\
                        0 & 0 & 0\\
                        -1 & 0 & 0
                        \end{bmatrix}x_k\\
\nonumber &~ +  \beta_{1k}
                        \begin{bmatrix}
                        1 & 0 & 0\\
                        0 & 0 & 0\\
                        0 & 0 & 0
                        \end{bmatrix}u_k
                        + \beta_{2k}
                        \begin{bmatrix}
                        0 & 0 & 0\\
                        0 & 1 & 0\\
                        0 & 0 & 0
                        \end{bmatrix}u_k
                        +d_k,
\end{align}
where $\bar{\alpha}_1=\bar{\beta}_1=0.05$, $\bar{\alpha}_2=\bar{\beta}_2=0.015$.
Let initial state variance matrix $X_0=I$, additive noise covariance matrix $D=0.5I$, and the weight matrices are selected as $Q=R=I$.
The exact solution to SARE \eqref{lem_op_2} is
\begin{align}
\nonumber P^{\ast}=  \begin{bmatrix}
                      1.5864  &  0.0673 &   0.1208\\
                      0.0673  &  1.4252 &   0.0528\\                      0.1208  &  0.0528 &   1.3770
                     \end{bmatrix}
\end{align}
and the optimal control gain is
\begin{align}
\nonumber L^{\ast}=\begin{bmatrix}
                    -0.5175 &  -0.0394 &  -0.0761\\
                    -0.0404 &  -0.4419 &  -0.0353\\
                    -0.0776 &  -0.0352 &  -0.4466
                     \end{bmatrix}.
\end{align}
Thus, one can obtain the optimal average cost $\lambda^\ast=  2.1943$ according to Lemma \ref{lem_op}.

Perform Algorithm \ref{A2} on system \eqref{sys1} with parameters:
$L^{(0)}=0$, $N=42000$, $\sigma_u^2=0.64$, $\varpi=10^8$,
$\tau_{max}=10$, and $\varepsilon=0.05$.
Algorithm \ref{A2} stops after four iterations and returns the estimated optimal control gain
 \begin{align}
\nonumber \hat{L}=
 \begin{bmatrix}
     -0.5192 &  -0.0443 &  -0.0795\\
   -0.0342  & -0.4437  & -0.0336\\
   -0.0747  & -0.0371  & -0.4497
 \end{bmatrix}
 \end{align}
  and  the estimated optimal average cost $\hat{\lambda}=2.2159$.

To make comparison, evaluate the following algorithms on system \eqref{sys1}:
(a) Algorithm 1 in this paper;
(b) MFLQv3 in \cite{pmlr-v89-abbasi-yadkori19a}, where the authors evaluated the value function  first and then Q-function. The iterative control policies were greedy with respect to the average of all previous Q-function estimations $\hat{Q}_1,\ldots,\hat{Q}_{i-1}$;
(c) Q-learning in \cite{bradtke1994adaptive}, where the kernel matrix $H^{(i)}$ was learned based on Bellman residual methods and RLS.

Instead of using the stop condition $\varepsilon$, we run the above three algorithms for 10 iterations.
The remaining parameter settings keep unchanged, except that for MFLQv3 in \cite{pmlr-v89-abbasi-yadkori19a}, in a single iteration, we use $21000$ and $21000$ time steps data to estimate the value function and Q-function, respectively.
The curves of $\|L^{(\tau)}-L^{\ast} \|$ and $\frac{|\lambda^{(\tau)}-\lambda^\ast|}{\lambda^\ast}$ are shown in Fig. 1 and Fig. 2, respectively.
From Fig 1 we can see that the control gain obtained by Algorithm \ref{A2} is closer to the optimal control gain than the other two algorithms.
Furthermore, Fig.2 shows that Algorithm \ref{A2} achieves much lower relative cost error.

\begin{figure}[hbt]
  \centering
  \includegraphics[width=.30\textwidth]{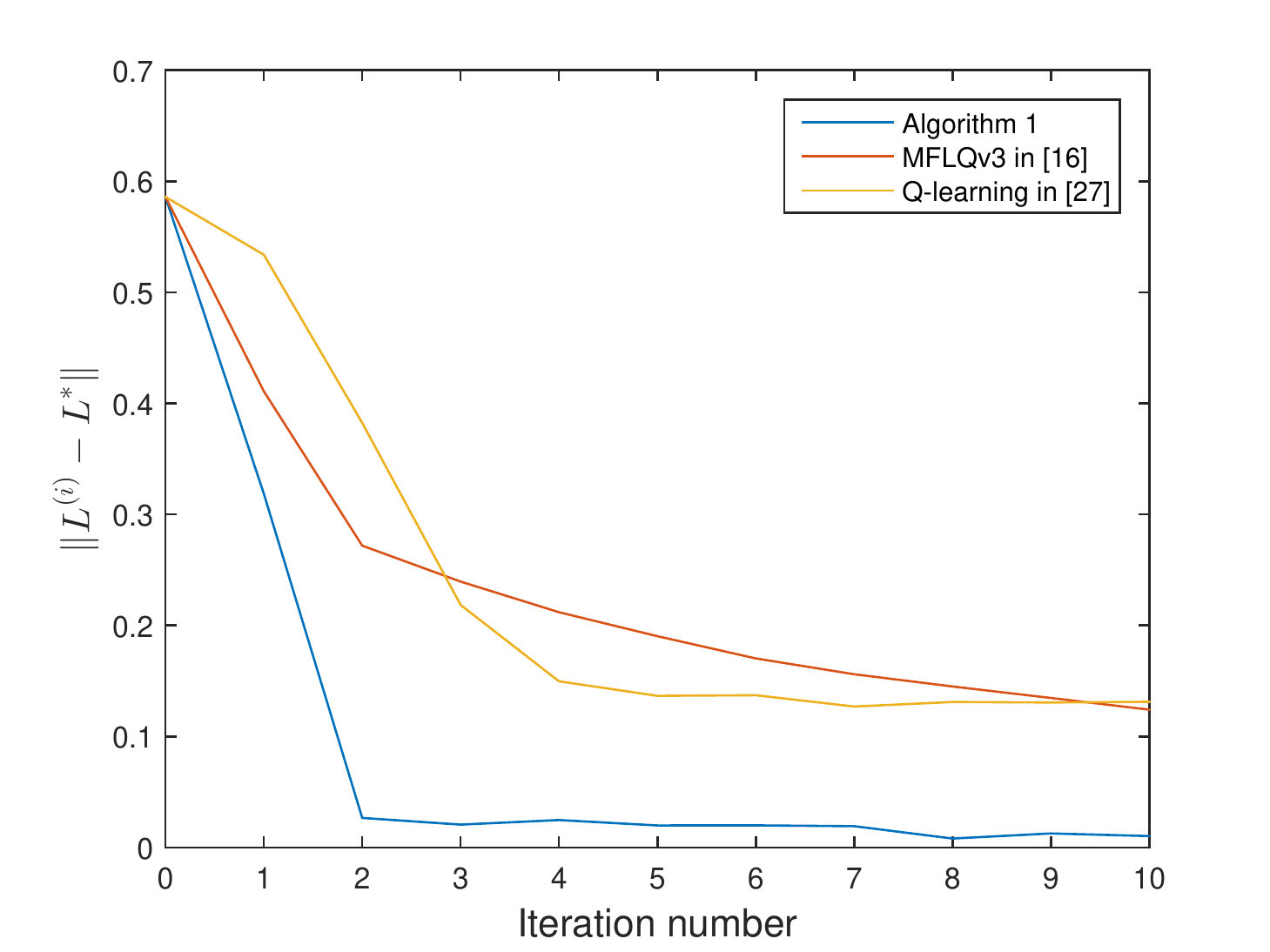}
  \caption{ The distance $\|L^{(\tau)}-L^{\ast}\|$ }
\end{figure}

\begin{figure}[hbt]
  \centering
  \includegraphics[width=.30\textwidth]{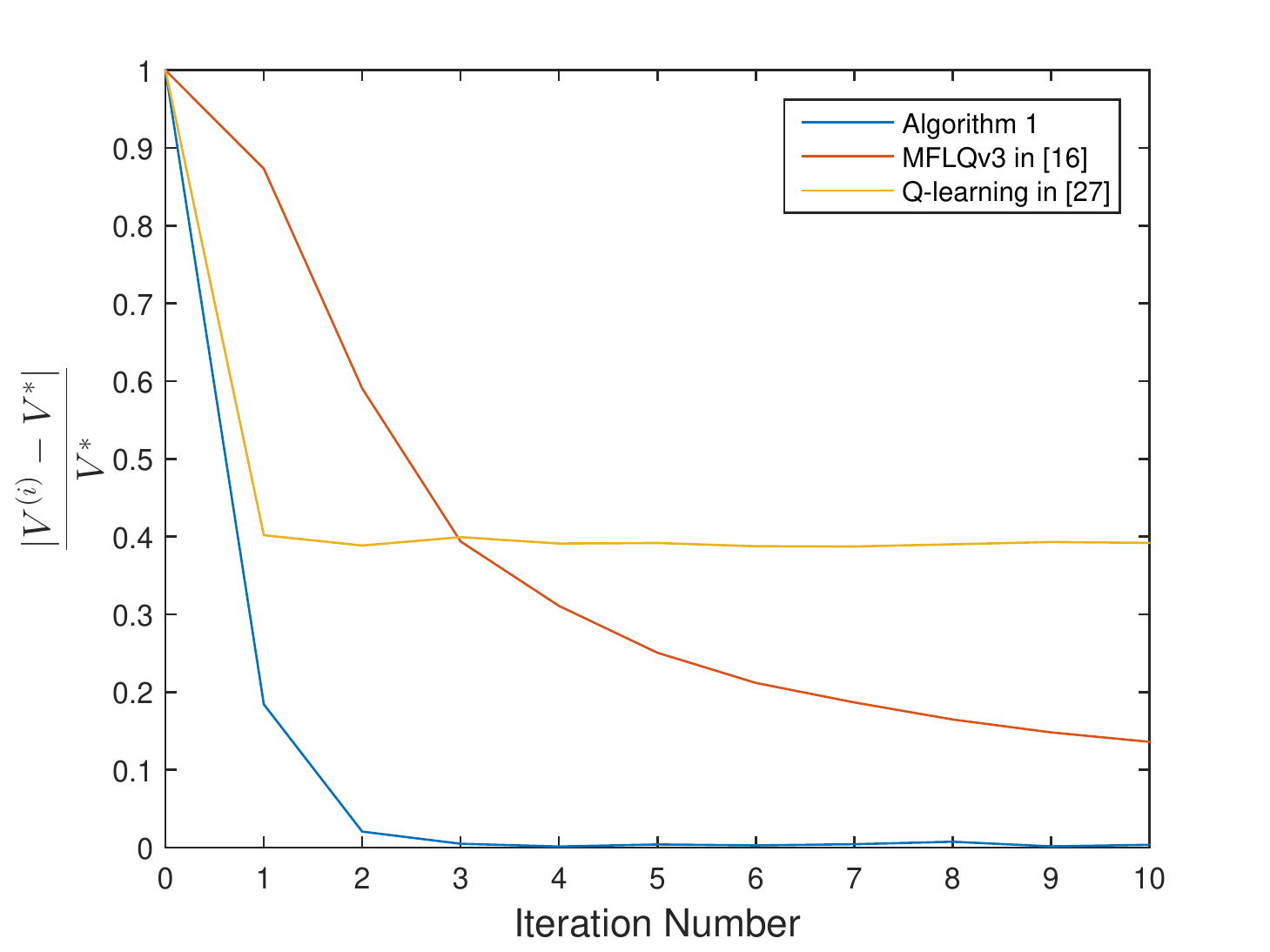}
  \caption{ Relative cost error $\frac{\|\lambda^{(\tau)}-\lambda^\ast\|}{\|\lambda^\ast\|}$}
\end{figure}

\section{conclusion}
This paper investigates the average cost optimal control problem for a class of discrete-time stochastic systems.
The system under consideration suffers from both multiplicative and additive noises.
Both model-based and model-free schemes are proposed to solve the stochastic LQR. The control policies and associated kernel matrices obtained from the two schemes are proved to converge to the optimal ones.
An model-free RL algorithm is presented to learn the optimal control policy in an online manner using the data of the system states and control inputs.
The proposed approach is illustrated through a numerical example.

\bibliographystyle{IEEEtranDoi}
\bibliography{mybib}

\end{document}